\documentclass[11pt]{article}
\usepackage[a4paper, total={6.67in, 10in}]{geometry}
\usepackage[T1]{fontenc}
\usepackage[utf8]{inputenc}
\usepackage{enumitem}
\usepackage{amsmath,amssymb}
\usepackage{amsthm} 
\usepackage{graphicx}
\usepackage[toc,page]{appendix}
\usepackage{float}
\usepackage{blindtext}
\usepackage[hyperfootnotes=false]{hyperref}
\usepackage[dvipsnames]{xcolor}
\usepackage{amsthm}
\usepackage{enumitem}

\usepackage{xcolor}

\newcount\Comments  

\newcommand{\kibitz}[2]{\ifnum\Comments=1{\color{#1}{#2}}\fi}

\usepackage{titlesec}
\titlespacing\section{0pt}{6pt plus 4pt minus 2pt}{6pt plus 2pt minus 2pt}
\titlespacing\subsection{0pt}{4pt plus 3pt minus 2pt}{4pt plus 2pt minus 2pt}
\titlespacing\subsubsection{0pt}{2pt plus 4pt minus 2pt}{2pt plus 2pt minus 2pt}
\definecolor{dg}{RGB}{2,101,15}
\newtheoremstyle{exampstyle}
  {4pt} 
  {4pt} 
  {\slshape} 
  {} 
  {\bfseries} 
  {.} 
  {.5em} 
  {} 
\theoremstyle{exampstyle}
\newtheorem{defin}{Definition}
\hypersetup{
    colorlinks=true,
    linkcolor=red,
    citecolor=DarkOrchid,
    filecolor=magenta,
    urlcolor=cyan
}
\setlist[itemize]{label=$\cdot$}
\usepackage[english]{babel}

\usepackage{natbib}
\usepackage{indentfirst}

\usepackage[english]{babel}
\usepackage{csquotes}
\usepackage{color, colortbl}
\definecolor{Gray}{gray}{0.9}
\definecolor{LightCyan}{rgb}{0.88,1,1}
\usepackage[first=0,last=9]{lcg}

\usepackage[dvipsnames]{xcolor}
\newcolumntype{b}{>{\columncolor{LightCyan}}c}
\newcolumntype{d}{>{\columncolor{Apricot}}c}

 \newcommand{\code}[1]{\textcolor{Green}{\texttt{#1}}}
 
\newtheorem{theorem}{Theorem}[section]

\newtheorem{lemma}[theorem]{Lemma}
\usepackage{algorithm}
\usepackage{algorithmic}
\usepackage{multirow}

\floatname{algorithm}{Procedure}
    \makeatletter
\def\@fnsymbol#1{\ensuremath{\ifcase#1\or \dagger\or \ddagger\or
   \mathsection\or \mathparagraph\or \|\or **\or \dagger\dagger
   \or \ddagger\ddagger \else\@ctrerr\fi}}
    \makeatother

\usepackage{titling}
\title{\textbf{Low-cost attacks on Ethereum 2.0 by \\sub-1/3 stakeholders}}
\date{\small{ School of Engineering and Applied Sciences\\Harvard University}}
\author{
Michael Neuder\thanks{First published in the Game Theory in Blockchain (GTiB) workshop at the 2020 Conference on Web and Internet Economics (WINE).}\thanks{\href{mailto:michaelneuder@g.harvard.edu}{michaelneuder@g.harvard.edu}}
\and Daniel J. Moroz\thanks{\href{mailto:dmoroz@g.harvard.edu}{dmoroz@g.harvard.edu}} 
\and Rithvik Rao\thanks{\href{mailto:rithvikrao@college.harvard.edu}{rithvikrao@college.harvard.edu}}
\and David C. Parkes\thanks{\href{mailto:parkes@eecs.harvard.edu}{parkes@eecs.harvard.edu}}}
\begin{document}
\maketitle
\vspace{-2.5em}

\begin{abstract}
We outline two dishonest strategies that can be cheaply executed on the Ethereum 2.0 beacon chain, even by validators holding less than one-third of the total stake: malicious chain reorganizations (``reorgs'') and finality delays. In a malicious reorg, an attacker withholds their blocks and attestations before releasing them at an opportune time in order to force a chain reorganization, which they can take advantage of by double-spending or front-running transactions. 
To execute a finality delay an attacker uses delayed block releases and withholding of attestations to increase the mean and variance of the time it takes blocks to become finalized. This impacts the efficiency and predictability of the system.   
We provide a probabilistic and cost analysis for each of these attacks, considering a validator with 30\% of the total stake. 
\end{abstract}

\section{Introduction}
Ethereum 2.0 represents the most substantial modification to the Ethereum protocol since its initial release. The first step of Ethereum 2.0, dubbed \textit{Phase 0}, is the creation of the {\em beacon chain}, which will run in parallel to the original Ethereum chain. The beacon chain fully implements a Proof-of-Stake consensus protocol, which is presented by \cite{gasper} and defined in the \cite{eth2spec}. The Proof-of-Stake protocol  is a combination of Casper FFG ``friendly finality gadget'' from \cite{buterin2017casper} and a modified version of the GHOST fork-choice rule of \cite{ghost}.

The security of the beacon chain relies on no attacking party controlling more than $1/3$ of the total stake, where the $1/3$ security threshold comes from Byzantine fault tolerance guarantees \citep{castro1999practical}. Further, the beacon chain is designed to penalize provably dishonest behavior (e.g., proposing two conflicting blocks at the same block height) through the use of ``slashing conditions'' (rules that lead to a loss of staked currency)~\citep{gasper}.
In this paper we nevertheless outline two dishonest strategies that  can be used by an adversary with 30\% of the total stake. Neither strategy involves cryptographically signing conflicting pieces of information, and thus neither is punishable through slashing conditions. The attacks are relatively cheap to execute. While the market capitalization of Ethereum is over 55 billion USD as of November 2020,  the  finality delay attack would cost on the order of 500 USD and the chain reorganization attack would cost on the order of 5 USD.

We first examine small-scale malicious chain reorganizations (``reorgs''). Reorgs are caused by a fork-choice rule deeming a new branch to be dominant over an existing branch, effectively deleting blocks from the canonical chain (the fork which honest participants identify as the head based on a fork-choice rule).  Reorgs can occur naturally due to network latency. In Proof-of-Work, for example, if two miners create blocks $A$ and $B$ at nearly the same time, the network partitions, with some honest nodes mining on each block. If $A'$ is a subsequent block that is mined as a child of $A$ before any blocks are mined as a child of $B$ then the  honest network will see $A'$ as the head of the chain. Block $B$ will be orphaned, and no longer considered part of the main chain. For an honest miner who previously saw $B$ as the head of the chain, switching to $A'$ has the effect of deleting block $B$ from the canonical chain.

 \textit{Malicious reorgs} are artificial reorgs that are caused by an attacker who seeks to exploit them by attempting to double-spend or front-run large transactions. A double-spend transaction is one where coins are sent to a counter-party in return for something of value, with the chain subsequentely reorged to delete the original transaction. Minimizing the feasibility of double-spending is critical to the security of cryptoeconomic systems, and  the original Bitcoin white paper  presents a probabilistic analysis of the feasibility of malicious reorgs~\citep{nakamoto2008bitcoin}. A desirable property is that the probability of a malicious reorg  decreases as the length of the reorg grows. This provides that transactions included in blocks that have many blocks built on top of them will remain on the chain  with high probability. This also comes with a trade off, namely long wait times for transactions to be deemed valid,   limiting 
the utility of the economic system.

Malicious reorgs also enable front-running~\citep{daian2019flash}, which is the process of using information about large transactions to create short-term arbitrage opportunities. Typically, front-running in Ethereum happens in decentralized exchanges (DEXs). Attackers pay high gas fees to ensure their transactions are included first in a new block, in this way ensuring they come before large transactions that they want to front-run. Malicious reorgs can be used to gain  fine-grained control on the order of   transactions.  For example, consider an attacker who knows they can reorg an upcoming block, $B$, with a block of their own, $A$. If a large buy order from another party to an exchange, $t_\ell$, is made and included in $B$, the attacker knows that the price on the exchange will increase, and thus  $A$ can usefully include a transaction buying from that exchange, $t_b$, followed by this large transaction $t_\ell$, followed by a transaction selling to that exchange, $t_s$.
If the attacker were playing honestly, they would have already released block $A$ by the time $t_\ell$ is heard over the network.

In addition to malicious reorgs, we demonstrate a strategy in which an attacker can delay 
the {\em finality} of the beacon chain. Finality is a property of blocks, and \cite{buterin2017casper} show that once a block is finalized, the block and the transactions it contains will only be removed from the canonical chain if the network is ``$1/3$-slashable'' (i.e., if validators controlling $1/3$ of the total stake are provably dishonest). In a healthy network, a block will become finalized within two epochs, where an epoch is a sequence of block creation opportunities with total duration of around 6.4 minutes (defined more formally in Section~\ref{sec:forkchoice}). We demonstrate that an attacker with a 30\% stake can delay finality with high frequency; e.g., delaying finality by three epochs on average every hour (a delay of around 19 minutes). This is a denial-of-service (DoS) attack on any transaction that needs to be finalized, and doing this with some regularity could severely damage the health and predictability of the network.

Since this paper appeared in the Game Theory in Blockchain workshop of the 2020 Conference on Web and Internet Economics (WINE), the Ethereum Foundation has proposed a change to the fork-choice rule to be included as part of the first hard fork on the beacon chain, HF1 \citep{hf1}, citing the reorg attack presented in this paper as motivating the change.


\subsection{Related Work}

The analysis in this paper is related to the Proof-of-Work selfish mining literature initiated by \cite{eyal2018majority}, and further expanded in \cite{nayak2016stubborn,kwon2017selfish,sapirshtein2016optimal}. Selfish mining and double-spending were extended to Proof-of-Stake (PoS) by \cite{brown2019formal}, who demonstrated that longest-chain PoS protocols with certain properties are vulnerable to reorgs. \cite{neuder2019selfish,neuder2020defending} demonstrate instantiations of these reorg attacks in the Tezos protocol.

A number of attacks on the Ethereum 2.0 beacon chain have also been presented. \cite{neu2020ebb} describe an attack in which the attacker exploits the expected time of release of attestations to partition the honest network, which is a different method of executing the finality delay than
the one presented here. \cite{decoyflipflip} presents an attack called the ``decoy-flip-flop'' where an attacker accumulates and releases conflicting attestations to ensure that epochs do not become finalized. This differs from our finality delay attack  in that the attacker that we describe does not accumulate old epoch attestations or sign conflicting attestations. \cite{pullreq-decoy-eth} seeks to mitigate the ``decoy-flip-flop'' by checking that an attestation's target epoch is the current or previous epoch, thus limiting the ability of the attacker to stockpile old attestations. Neither of our attacks rely on sending attestations from old epochs and they are not resolved by this fix. \cite{bouncingattack} presents a ``bouncing attack'' where an attacker uses the fact that the fork-choice rule only operates on blocks since the last justified epoch (see Section~\ref{sec:subsec-fin} for definition of justification) to take advantage of network latency issues. \cite{pullreq-bounce-eth} aims to prevent this attack by ensuring that the last justified epoch does not change after the first two slots of an epoch. Neither of the attacks we present rely on changing the view of the last justified epoch, so this fix has no effect on them.

\section{Ethereum 2.0 Proof-of-Stake}\label{sec:ethereum}

The Ethereum 2.0 consensus mechanism is best understood in two parts: the {\em fork-choice rule} (HLMD-GHOST\footnote{Hybrid Latest Message Driven Greedy Heaviest Observed SubTree}) and the {\em finality gadget} (Casper FFG\footnote{Friendly Finality Gadget}), as described by \cite{gasper}. The fork-choice rule defines which blocks are considered to be on the canonical chain, while the finality gadget defines which of the produced blocks are \textit{justified} or \textit{finalized} (defined in Section~\ref{sec:subsec-fin}).

\subsection{Block Creation and the Fork-Choice Rule}\label{sec:forkchoice}
In the \cite{eth2spec}, time is partitioned into \textit{epochs}, which are further subdivided into 32 \textit{slots} of 12 seconds each. Each slot will contain a \textit{committee} with at least 128 validators. The first (randomly selected) validator in each committee takes the role of the \textit{proposer} for that slot, and can create a block. Epoch boundary blocks (EBBs)\footnote{These are defined as \textit{epoch-block pairs} in \cite{gasper} because they account for the possibility that an EBB may be used for multiple epochs, but for the scope of this work, just using EBBs is sufficient.} are checkpoint blocks that represent the start of an epoch based on the current set of blocks that a validator has observed.

\begin{defin}\label{def:ebb}
A validator determines the \underline{epoch boundary block} (EBB) for epoch $i$ as follows:
\begin{enumerate}[noitemsep,topsep=4pt]
    \item Choose the block at the first slot in epoch $i$ if it is present in the validator's canonical chain
    \item Else, choose the highest slot block on the canonical chain from epoch $i'$ such that $i' < i$. 
\end{enumerate}
\end{defin}
During each slot, all validators in the committee can issue a single \textit{attestation}. 
\begin{defin}\label{def:attestation}
An \underline{attestation} is the casting of a vote that contains:
\begin{itemize}[noitemsep,topsep=4pt]
    \item[$A_1$] $-$ a source EBB
    \item[$A_2$] $-$ a target EBB
    \item[$A_3$] $-$ the head of the chain according to HLMD-GHOST
\end{itemize}
Let $(A_1=\beta,A_2=\gamma,A_3=\delta)$ denote an attestation for blocks $\beta, \gamma$,  and $\delta$. 
\end{defin}
The $A_1$ and $A_2$ components of attestations are defined as the \textit{source} and \textit{target} EBBs of the Casper FFG vote. For example,  attestation $(A_1=0,A_2=32)$ can be interpreted as ``I want to move the finality gadget from the block at slot 0 to the block at slot 32''. We discuss how these votes contribute to the finalization process in Section~\ref{sec:subsec-fin}. For simplicity, we will sometimes elide part of an attestation when it is not relevant to the analysis; e.g., writing an attestation as $(A_3=\delta)$ when  the values of $A_1$ and $A_2$ are unimportant.

\begin{defin}\label{def:lmd-ghost-weight}
Let the \underline{weight} of a given block, B, be the sum of the stake owned by validators who attest (with $A_3$) to B or any descendants of B.
\end{defin}
\begin{defin}\label{def:lmd-ghost}
Based on the most recent set of attestations, one per validator, \underline{LMD-GHOST} \footnote{Latest Message Driven Greedy Heaviest Observed SubTree} dictates that the head of the chain can be reached by starting at the genesis block and choosing the higher weight branch at each fork until hitting a leaf block (a block with no children). 
\end{defin}

The hybrid version of Definition~\ref{def:lmd-ghost} from \cite{gasper} also adds that the start of the LMD-GHOST weight calculation is the last justified EBB (defined in Section~\ref{sec:subsec-fin}) rather than the genesis block, which improves the computational efficiency of the system. Our attacks are not impacted by this distinction because they do not rely on any attestations from before the current epoch.

For the remainder of this work, we  assume that each validator has equal stake, which we normalize to 1 without loss of generality. This allows us to count attestations instead of summing the weight over each attesting validator.
Figure~\ref{fig:ghost} demonstrates the use of HLMD-GHOST to pick the head of the canonical chain. Based on the results of HLMD-GHOST, an honest validator can choose which block will receive their $A_3$ vote.
In an  network with full participation, with honest validators, and with no network delays, the proposer releases a block at the start of the slot and each validator in that slot's committee votes for this block as the head of the chain.

\begin{figure}
    \centering
    \includegraphics[width=0.7\linewidth]{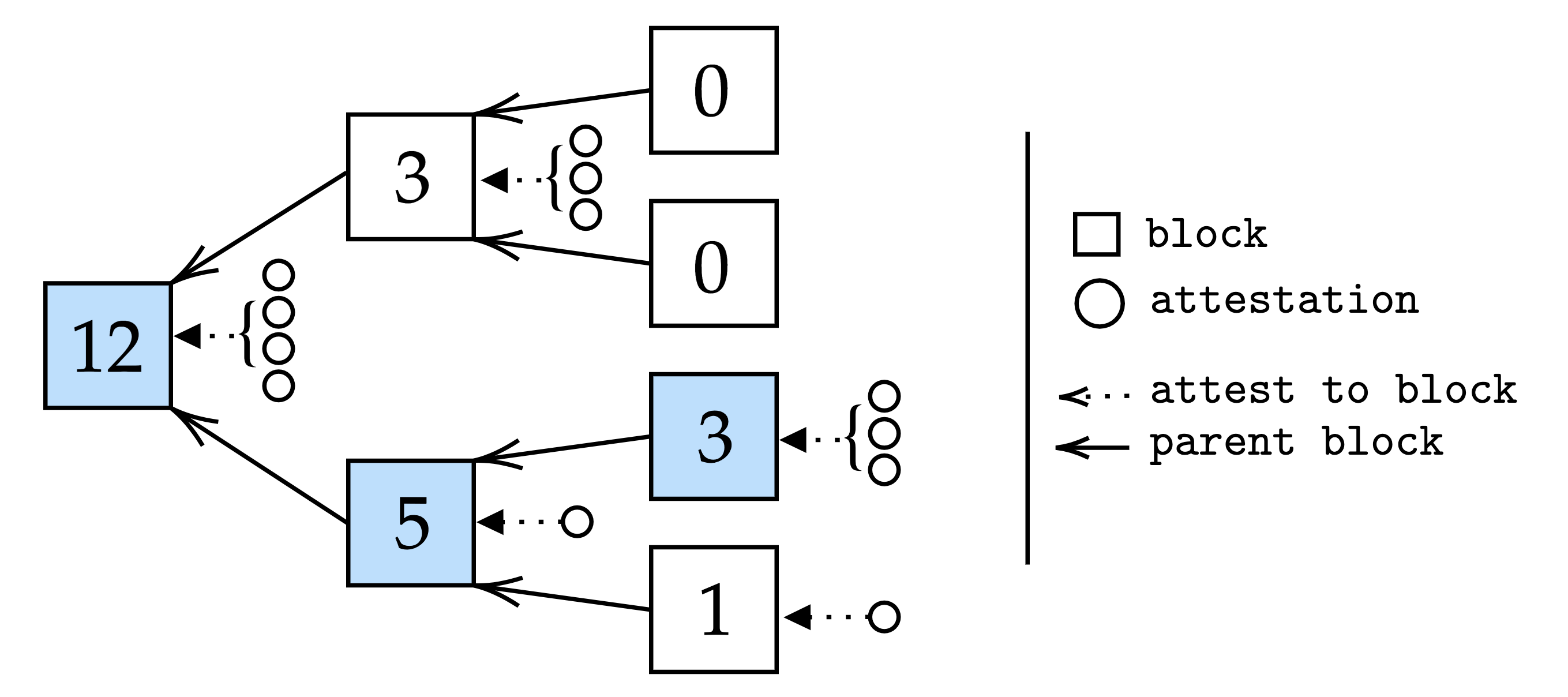}
    \caption{The results of running HLMD-GHOST to choose one of four conflicting leaf blocks. Each block is annotated with its weight, which is the number of attestations to itself or a descendant block. The blue blocks are the heaviest at each fork and thus part of the canonical chain. The blue block annotated with `3' is selected as the head of the chain.
    \label{fig:ghost}}
\end{figure}

\subsection{Finality Gadget}\label{sec:subsec-fin} 

Finality is achieved through the use of the source and target EBBs (the $A_1$ and $A_2$ components of attestations from Definition~\ref{def:attestation}).
\begin{defin}
There is a \underline{supermajority link between EBBs $\beta$ and $\gamma$} if validators controlling $2/3$ of the total stake release attestations with $(A_1=\beta,A_2=\gamma)$. 
\end{defin}

\begin{defin}\label{defin:justified}
The genesis (first) block is defined to be \underline{justified}. An EBB, $\gamma$, is \underline{justified} if there is a supermajority link between a previously justified EBB, $\beta$, and $\gamma$.
\end{defin}
%
\begin{figure}
    \centering
    \includegraphics[width=0.9\linewidth]{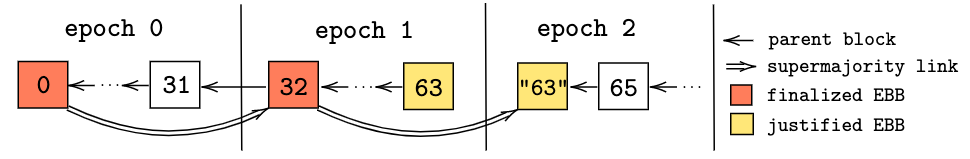}
    \caption{The process of justification and finalization for the first three epochs. In epoch 1, a supermajority link with $(A_1=0,A_2=32)$ justifies block 32. Since the EBB of epoch 2 was not produced as expected (there is no block 64), block 63 is ``borrowed'' from the previous epoch to fill in as the EBB for epoch 2. The supermajority link with $(A_1=32,A_2=63)$ finalizes block 32 and justifies block 63.
    \label{fig:justification}}
\end{figure}
\begin{defin}\label{defin:finalized}
A justified EBB, $\beta$, is \underline{finalized} if there is a supermajority link between $\beta$ and $\gamma$, and $\gamma$ is the justified EBB of the next epoch.\footnote{This is Case 4 of the ``Four-Case Finalization Rule'' in Section 8.5 of \cite{gasper}. Though the other 3 cases are implemented in the spec, they do not impact our  analysis  and are omitted for brevity.} 
\end{defin}
 
In other words, EBBs become justified by being the target of a supermajority link, and become finalized by being the source of a supermajority link from the previous epoch.
For example, consider the first three epochs pictured in Figure~\ref{fig:justification}. The genesis block at slot 0 is justified by definition. Over the course of epoch 1, validators controlling $2/3$ of the total stake issue attestations with $(A_1=0,A_2=32)$, which creates a supermajority link and justifies $32$. Subsequently in epoch 2, a supermajority link is established with $(A_1=32,A_2=63)$, which finalizes block 32 and justifies block 63. Block 63 is used as the EBB of epoch 2 because the expected EBB, block $64$, was not created (Definition~\ref{def:ebb}). 

Once an EBB is finalized, all the blocks in the corresponding epoch and any prior epochs are also finalized. \cite{buterin2017casper} demonstrate that these blocks will remain on the canonical chain permanently unless validators controlling at least $1/3$ of the total stake are provably dishonest.

\section{Malicious Reorgs}\label{sec:malreorgs}

In this section we present a dishonest strategy in which the attacker can force malicious reorgs by privately creating blocks and releasing their private fork at an opportune time. We analyze the case where the attack takes place within a single epoch.

\begin{figure}
    \centering
    \includegraphics[width=0.7\linewidth]{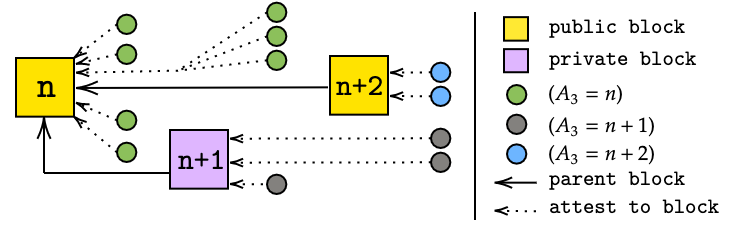}
    \caption{An example of using a 1-block private fork to execute a length-1 malicious reorg. An attacker can attest to its private $n+1$ block using both its slot $n+1$ and $n+2$ attestations. Thus, a minority attacker's attestations on block $n+1$ can outnumber the honest attestations on block $n+2$. Before block $n+1$ is released by the attacker, the $n+2$ block is the head of canonical chain based on HLMD-GHOST. After block $n+1$ is released, the honest nodes see that the weight of block $n+1$ is 3 (the number of grey circles) and the weight of the block $n+2$ is 2 (the number of blue circles). As a result, HLMD-GHOST denotes $n+1$ as the head of the chain because it is the heaviest leaf block that is a child of block $n$, and block $n+2$ becomes an orphan. 
    \label{fig:reorg}}
\end{figure}

\begin{defin}
In a \underline{length-$n$ malicious reorg},  $n$ consecutive blocks that are part of the canonical chain are orphaned as a result of  a reorganization
that is caused by an attacker.
\end{defin}

Note that all attestations in this section refer to $A_3$, or the fork-choice component of the attestation.

\subsection{Strategy}

Consider Figure~\ref{fig:reorg}, in which the attacker (illustrated with the grey attestations and purple blocks) is the block proposer for the $n+1$st slot. The attacker executes a length-$1$ malicious reorg to orphan the $n+2$nd slot block in the following manner:
\begin{enumerate}[noitemsep,topsep=4pt]
    \item The attacker privately creates block $n+1$ and uses its slot $n+1$ attestations to privately attest that block $n+1$ is the head of the chain, that is $(A_3=n+1)$. As the honest attesters in the $n+1$st committee will not see the private block, they instead attest with $(A_3=n)$.
    \item At slot $n+2$, an honest validator will propose a block whose parent is the slot $n$ block because it has not observed block $n+1$. The honest slot $n+2$ committee members attest $(A_3=n+2)$. The attacker uses its slot $n+2$ attestations to privately attest $(A_3=n+1)$.
    \item The attacker then releases block $n+1$ and the private attestations. Since both blocks $n+1$ and $n+2$ are leaf blocks with a common parent, they are conflicting and HLMD-GHOST must choose one. Block $n+1$ has attacker attestations from both slots $n+1$ and $n+2$, while block $n+2$ has attestations from honest validators from slot $n+2$ only (since honest $n+1$ attestations were given to the common parent block $n$). In this example, block $n+1$ has a higher weight than $n+2$ (3 versus 2 respectively), and thus is seen as the head of the chain by HLMD-GHOST. This causes the $n+2$ block to be orphaned.
\end{enumerate}

This attack is possible because during the private fork, the honest attestations of slot $n+1$ will vote for the block at slot $n$, which is an ancestor of block $n+1$ and thus isn't a competing fork.
Figure \ref{fig:reorg} demonstrates a relatively minor reorg, in which only a single honest block is deleted from the chain, but this pattern can be extended to execute longer attacks. See Appendix \ref{appendix:malicious-reorgs} for calculations in regard to the frequency and cost of this malicious reorg attack. 

The left-hand plot of Figure~\ref{fig:prob_reorg} shows numerical approximations to the probability that a malicious reorg is feasible for a 30\% stake attacker as a function of the reorg length $n$, based on a Monte Carlo simulation of $10^7$ randomly generated epochs. While long reorgs are uncommon, it is possible for a 30\% attacker to frequently execute short reorgs. As an example,  a length-$3$ reorg can be executed by a 30\% attacker  once an hour for a total cost of 2 USD per reorg. From the right-hand plot of Figure~\ref{fig:prob_reorg}, we see that the attacks are cheap; a length $n$ reorg costs the attacker approximately $n-1$ USD. 

\begin{figure}
    \centering
    \includegraphics[width=0.8\linewidth]{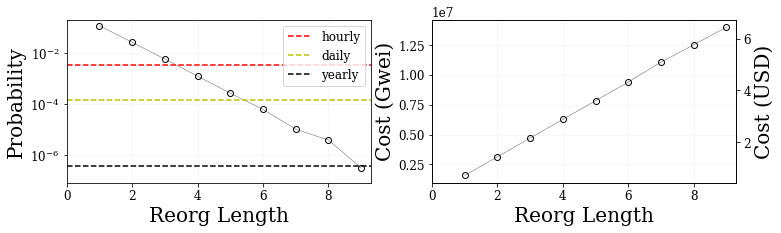}
    \caption{The probability and cost of a malicious reorg as a function of the reorg length for a 30\%-stake attacker. These approximations are from a Monte Carlo simulation of $10^7$ epochs. The cost is shown in Gwei ($10^9$ Gwei = $1$ Eth) as well as USD.
    \label{fig:prob_reorg}}
\end{figure}

\section{Finality Delay}\label{sec:delayfin}

In this section we present a dishonest strategy in which an attacker seeks to delay the finality of future epochs by delaying the release of EBBs.

\begin{defin}
In a \underline{length-$n$ finality delay}, an attacker ensures that none of the next $n$ epochs are finalized on time.
\end{defin}
This attack ensures that over the next $n$ epochs, no new transactions are finalized. This does not mean  that these transactions will never be finalized. Since an epoch being finalized also finalizes any blocks in previous epochs, once the attack  ends, a future epoch  becomes finalized (likely the second after the attack ends) and all the transactions sent during the attack will be as well. This is why we refer to the attack as a finality \textit{delay} rather than a permanent liveness attack on Casper FFG.

Note that attestations in this section may refer to $A_1, A_2$ (Casper FFG vote), or $A_3$ (HLMD-GHOST vote), and thus we will specify this as necessary.

\begin{figure}
    \centering
    \includegraphics[width=0.75\linewidth]{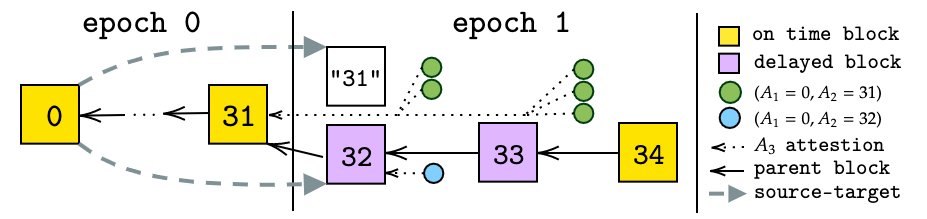}
    \caption{An attacker who makes use of a  delayed release of the slot 32 and slot 33 blocks to prevent epoch 1 from being justified. The 5 honest attestations (green circles) over the course of slot 32 and 33 will identify block 31 as the EBB of epoch 1. Once the attacker releases blocks 32 and 33, the honest validators will recognize block 32 as the EBB, but with the attacker withholding their remaining attestations a supermajority link with $(A_1=0, A_2=32)$ is not formed and epoch 1 is never justified. 
    \label{fig:fin}}
\end{figure}
\subsection{Strategy}
 
An attacker can delay finality by being the proposer of the EBB and the subsequent block (i.e., block 32, the EBB of epoch 1, and block 33 in Figure~\ref{fig:fin}), and creating a private fork until there is enough attestations from other validators that a prior epoch block is the target EBB (Definition~\ref{def:ebb}). In Figure~\ref{fig:fin}, the other validators make attestations with $(A_1=0,A_2=31)$ where block 31 is the target EBB.
When the private fork is revealed by the attacker, block 32 is recognized as the start of  epoch 1 but because Casper FFG votes for block 31 have been incorrectly cast and the attacker withheld their own attestations, block 32 cannot gain sufficient attestations. No supermajority link between either $0$ and $32$ or $0$ and $31$ is formed, and  epoch 1 is not justified. 

In the context of Figure~\ref{fig:fin}, this attack unfolds as follows: 
\begin{enumerate}[noitemsep,topsep=4pt]
    \item Consider an attacker who is the proposer for block 32 (the EBB of epoch 1) as well as for block 33. After block 31 is published, the attacker creates block 32, a single attestation $(A_1=0,A_2=32,A_3=32)$, and block 33. These are all kept private until step 3.
    \item Since the honest validators do not see the EBB (block 32), they attest with $(A_1=0,A_2=31,A_3=31)$ that the last published block, 31, is the highest slot block of the previous epoch, and thus the new EBB (by Definition~\ref{def:ebb}). In order to succeed in the attack, a 30\% attacker waits until 139 (about 3.3\% of the total number of attestations in the epoch) honest attestations are published that identify 31 as the FFG target block (see Appendix \ref{appendix:finality-delay}).
    \item The attacker  releases the private fork, which is  identified as the head by HLMD-GHOST  because 31 is the predecessor of block 32 and therefore the honest attestations for $(A_3=31)$ do not conflict with the attacker's attestation to $(A_3=32)$. The effect is that the  honest attestations are wasted on the wrong EBB for epoch 1.  Since 33 is the only leaf block, it is seen as the head of the chain.
    \item The attacker withholds the rest of their attestations in epoch 1 that would otherwise indicate that 32 is the correct EBB. No supermajority link with $(A_1=0,A_2=32)$ or $(A_1=0,A_2=31)$ is created, and epoch 1 is not justified.
\end{enumerate}

See Appendix \ref{appendix:finality-delay} for calculations of the cost of this attack and an example calculation of how many attestations a 30\% stake attacker needs to ``waste'' (cause to incorrectly identify the target EBB, $A_2$) for the attack to succeed.
From Figure~\ref{fig:fin}, the 30\% attacker has a $(0.3)^2 = 0.09$ probability of ensuring a non-justified epoch. This is the probability that the attacker is the block proposer for the EBB (block 32) as well as the subsequent block.
Since finality occurs when two consecutive epochs are justified, the attacker can  delay finality for as long as they can ensure no two epochs in a row become  justified.

\subsection{Probability of Finality Delays}
The probability of a length-$n$ delay in finality is equal to the  probability that there are no two heads in a row in a sequence of $n$ flips of a biased coin with $P(H) = 0.91$ and $P(T) = 0.09$. We calculate this probability through enumeration.
%
The left-hand plot of Figure~\ref{fig:prob_fin} shows this probability for different values of $n$; a 30\% attacker can delay finality with high frequency. The right-hand plot of Figure~\ref{fig:prob_fin} gives the cost of the attacks, which ranges from 100 to 1200 USD.
\begin{figure}
    \centering
    \includegraphics[width=0.8\linewidth]{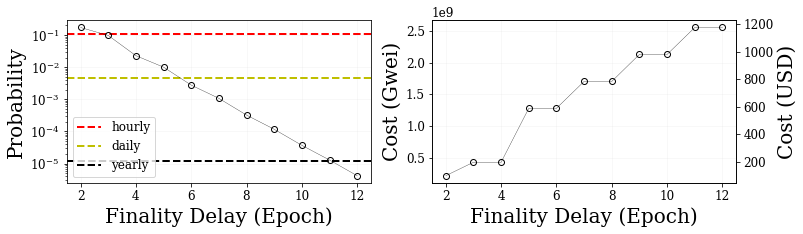}
    \caption{The probability  with which a 30\% attacker can delay finality for $n$ epochs, and the associated cost. Hourly, daily, and yearly rates are delineated by horizontal lines. See Appendix~\ref{appendix:fincost} for the calculation of the cost.}
    \label{fig:prob_fin}
\end{figure}

\section{Conclusion}

We have presented two dishonest strategies that a 30\% stake attacker on the Ethereum 2.0 beacon chain can launch with high probability and low cost. 
Though the focus of this paper is on introducing these two strategies and studying them for a 30\% attacker, it is also worth considering the attack possibilities for participants with smaller stakes.
In both cases, the likelihood of an attack succeeding decreases  with the available stake of an attacker. For example, we can show that a 25\% attacker can execute a length-2 reorg once per 3 hours and a 2 epoch finality delay once per hour.
%
%
An interesting avenue for future work is studying how these probabilities change as a function of stake. 
Further, it would be useful to quantify how much an attacker could expect to earn by executing a malicious reorg and how much a  finality delay attack can impact 
the average transaction finalization times.

\section*{Acknowledgements}

The authors would like to thank an anonymous reviewer for useful comments as we prepared this manuscript for the Game Theory in Blockchain Workshop (GTiB) at the 2020 conference on Web and Internet Economics (WINE 2020). This work is supported in part by two generous gifts to the Center for Research on Computation and Society at Harvard University, both to support research on applied cryptography and society.

\bibliography{refs}

\appendix

\section{Calculating Probability and Cost of Malicious Reorgs}
\label{appendix:malicious-reorgs}

\subsection{Probability}
We consider the probability with which a 30\% attacker can execute a length-$n$ reorg within an epoch. Probabilistically, this can be phrased as follows. 

\begin{quote}
    \emph{Consider an urn with $128 \cdot 32 = 4096$ balls\footnote{128 represents the number of validators per committee, and 32 is the number of slots in each epoch.}, 30\% of which are red (the attacker) and the remainder are black. The balls are drawn in groups of $128$ without replacement until the urn is empty. What is the probability that there is a sequence of $m + n$ draws of 128, over which there are more red balls drawn in total than black balls drawn in the final $n$ draws of the sequence, for any $m$ where the first ball in each of the $m$ draws is red.} 
\end{quote}
Let $m$ be the length of the attacker fork, where each of the private blocks must be in consecutive slots, and $n$ be the number of honest blocks that will be orphaned. Intuitively, we care about the relationship between the number of attacker attestations over $m+n$ slots (grey circles in Figure~\ref{fig:reorg}) versus the number of honest attestations over the last $n$ of those slots (blue circles in Figure~\ref{fig:reorg}), because the honest attestations during the attacker fork (seen as green circles in Figure~\ref{fig:reorg}) will contribute to the weight of an ancestor to both conflicting forks, and thus won't impact the leaf block HLMD-GHOST count. We approximate this probability with Monte Carlo simulations.

\subsection{Cost}
All constants are shown in this \code{FONT} and are copied from the \cite{eth2spec}. Validators are rewarded for both creating attestations and including attestations in the blocks they create. Notice that in Figure~\ref{fig:reorg}, the attacker is still able to include the honest attestations heard over the private fork, because they attest to the predecessor block ``n''. Thus we only consider the loss in rewards to the attackers attestations.
Assume there are $32 \cdot 128$ validators who each have the \code{MAX\_EFFECTIVE\_BALANCE} of $32 \times 10^9$ Gwei $= 32$ Eth. 
\begin{defin}
The base reward (in Gwei), denoted $\rho$, for each validator is,
\begin{align}
    \rho &= \frac{\code{BASE\_REWARD\_FACTOR}}{\code{BASE\_REWARDS\_PER\_EPOCH}} \cdot \frac{32 \times 10^9}{\lfloor \sqrt{32 \cdot 128 \cdot 32 \times 10^9} \rfloor} \\ 
    &= 2^4 \cdot \frac{32 \times 10^9}{11448668} \approx 44721.
\end{align}
\end{defin}
\begin{defin}
The inclusion reward, denoted $\iota$, as a function of the inclusion delay (difference between the slot that the attestation is created for and the slot at which that attestation is included), $d$, for each attestation is,
\begin{align}
    \iota(d) &= \frac{1}{d} \left(\rho - \frac{\rho}{\code{PROPOSER\_REWARD\_QUOTIENT}} \right)\\ 
    &= \frac{\rho}{d} \left(1 - \frac{1}{2^3}\right) = \frac{7 \rho}{8d}.
\end{align}
\end{defin}

\begin{defin}
The max value, $\nu$, of a single attestation is,
\begin{align}
    \nu &= 3 \rho + \iota(1) \\
    &= \frac{31}{8} \cdot 44721 \approx 173294.
\end{align}
The attester is rewarded for being included, having the correct source and destination epoch boundary blocks, and pointing to the correct head of the chain. 
\end{defin}

When the attacker is following the reorg policy in Section~\ref{sec:malreorgs}, they only miss out on the inclusion reward for their attestations not being included immediately (because they will correctly identify source and target epoch boundaries as well as the head of the chain). 
Since the attacker attestations for blocks \textit{after the end of the private fork}\footnote{Attacker attestations during the private fork will be included by subsequent attacker blocks.} may not be included for some time (due to the fact that \code{MAX\_ATTESTATIONS}=128), we will assume that those attestations \textit{will eventually be included},\footnote{They will be included because the honest nodes will be incentivized to include as many attestations as possible due to the proposer rewards.} but that the inclusion reward will be zero because the value of $d$ will be high. Thus the cost of a length-$n$ malicious reorg is $k\cdot7\rho/8$, where $k$ is the number of attestations that the attacker creates after the end of the private fork.\footnote{Even though the attacker attestations may not immediately be included in the chain, they will still count towards the value of the LMD-GHOST fork choice.} Figure~\ref{fig:prob_reorg} demonstrates the low costs of reorg attacks.

\section{Finality Delay}
\label{appendix:finality-delay}

\subsection{Example of Finality Delay with 30\% Attacker}
Consider an attacker with 30\% of the total stake. In order to ensure a specific epoch is not justified, the attacker must ensure the $2/3$ threshold for justification isn't met. Given 32 slots per epoch and 128 validators per committee, there are 4096 attestations for each epoch. This implies that in order for an epoch to become justified, at least $\lceil 4096 \cdot 2/3\rceil = 2731$ attestations must identify the correct epoch boundary block, and thus the attacker needs to ensure that at least 1366 of the attestations are incorrect. The 30\% attacker will withhold their own $\lfloor 4096 \cdot 0.3\rfloor = 1228$ attestations, and thus needs to ensure $1366 - 1228 = 138$ honest attestations are incorrect.\footnote{We use 139 in the main text because the attacker creates a single attestation to their own EBB.} Because there are $128$ attestations per slot, this corresponds to the attacker needing at least 2 slots at the beginning of the epoch.

\subsection{Calculating Cost of Finality Delay}\label{appendix:fincost}
We now directly calculate the cost for each finality delay attack. Again we consider the case where there are $128 \cdot 32$ validators who each have the \code{MAX\_EFFECTIVE\_BALANCE} $=32 \times 10^9$ Gwei. Recall that in this case the max value, $\nu$, of a single attestation is $173294$ Gwei. When the network has not been finalized for 4 epochs, the system enters an inactivity leak state, where it remains for each subsequent epoch until finality  
\begin{defin}
The inactivity leak penalty, denoted $\lambda$, as a function of the number of epochs since finality, $e_f$, for $e_f >4$ is,
\begin{align}
    \lambda(e_f) &= e_f\cdot \frac{\code{MAX\_EFFECTIVE\_BALANCE}}{\code{INACTIVITY\_PENALTY\_QUOTIENT}}\\ 
    &= e_f \cdot \frac{32\times10^{9}}{2^{24}} \approx 1907.35e_f. 
\end{align}
\end{defin}
In addition to the inactivity leak penalty, the attacker will incur a penalty of $\nu$ for each missing attestation.\footnote{This penalty is in place to ensure optimal behavior is only giving a neutral balance.} Thus for each attestation that the attacker owns, the cost is $2\nu$, because if they were playing the honest strategy they would earn  $1\cdot\nu$, and now each attestation missing will be worth $-1\cdot\nu$ 

\begin{lemma}
The cost of a length-$n$ delay in finality, $C(n)$, for a 30\% attacker is, 
\begin{align}
    C(n) &= \begin{cases}
    \lceil n / 2 \rceil \cdot 1228 \cdot \nu & \text{if } n \leq 4 \\
    \lceil n / 2 \rceil \cdot 1228 \cdot 2\nu + \lambda(n) & \text{if } n > 4 \\
    \end{cases}
\end{align}
\end{lemma}
\begin{proof}
Each time the attacker denies the justification of an epoch, they miss out on all 1228 of their attestation rewards, which are each worth $\nu$. In the case where $n\leq 4$, no inactivity leak has begun, and the attacker only is losing their attestation rewards. In order to delay finality for $n$ epochs, they must at least ensure that every other epoch is \textit{not} justified, thus they must be paying the $1228\nu$ penalty $\lceil n/2 \rceil$ times. When $n>4$, the inactivity penalty kicks in and the attacker loses an additional $\lambda(n)$ and is also penalized $\nu$ for each attestation, making the total cost for their missing attestations over the course of an epoch $1228 \cdot 2\nu$.
\end{proof}

\end{document}